\documentclass[onecolumn,showpacs]{revtex4-2}
\usepackage{color}
\usepackage{xcolor}
\usepackage{graphics}
\usepackage{graphicx}
\usepackage{tikz-cd}
\usepackage{dcolumn}
\usepackage{amsmath}
\usepackage{amssymb}
\usepackage{mathtools}
\usepackage{hyperref}

\newtheorem{theorem}{Theorem}[section]

\newtheorem{proposition}[theorem]{Proposition}

\newtheorem{lemma}[theorem]{Lemma}
\newtheorem{definition}[theorem]{Definition}
\newenvironment{proof}[1][Proof]{\begin{trivlist}
\item[\hskip \labelsep {\bfseries #1}]}{\end{trivlist}}
\allowdisplaybreaks

\newenvironment{remark}[1][Remark]{\begin{trivlist}
\item[\hskip \labelsep {\bfseries #1}]}{\end{trivlist}}

\newcommand{\qed}{\nobreak \ifvmode \relax \else
	\ifdim\lastskip<1.5em \hskip- \lastskip
	\hskip 0.5em plus0em minus0.5em \fi \nobreak
	\vrule height0.75em width0.5em depth0.25em\fi}

\begin{document}

\title{On the existence of conformal Killing horizons in LRS spacetimes}

\author{Abbas M \surname{Sherif}}
\email{abbasmsherif25@gmail.com}
\affiliation{Department of Science Education, Jeju National University, Jeju, 63243, South Korea}

\begin{abstract}
Let $M$ be a locally rotationally symmetric spacetime, and $\xi^a$ a conformal Killing vector for the metric on $M$, lying in the subspace spanned by the unit timelike direction and the preferred spatial direction, and with non-constant components. Under the assumption that the divergence of $\xi^a$ has no critical point in $M$, we obtain the necessary and sufficient condition for $\xi^a$ to generate a conformal Killing horizon. It is shown that $\xi^a$ generates a conformal Killing horizon if and only if either of the components (which coincide on the horizon) is constant along its orbits. That is, a conformal Killing horizon can be realized as the set of critical points of the variation of the component(s) of the conformal Killing vector along its orbits. Using this result, a simple mechanism is provided by which to determine if an arbitrary vector in an expanding LRS spacetime is a conformal Killing vector that generates a conformal Killing horizon. In specializing the case for which $\xi^a$ is a special conformal Killing vector, provided that the gradient of the divergence of $\xi^a$ is non-null, it is shown that LRS spacetimes cannot admit a special conformal Killing vector field, thereby ruling out conformal Killing horizons generated by such vector fields.
\end{abstract}

\maketitle

\section{Introduction}\label{sec1}

Understanding black holes has taken an integral role in various areas of theoretical physics, including general relativity (GR) and quantum gravity \cite{hawe1,rw1}. Many of their fundamental properties have been uncovered from the boundary (horizon) that demarcates timelike regions of physical observers from the interior of the black hole.  The \textit{event horizon} of static solutions has provided a rich background for the exploration quantum effects of black holes, and many of their properties were largely understood a long time ago for various spacetimes. 

There are however two drawbacks with the notion of event horizons. Firstly, it is of a global character: one requires the full knowledge of the future asymptotic region of the spacetime to locate the horizon. The second issue is of a teleological nature, where the event horizon forms in anticipation of an impending collapse (this phenomenon is exhibited in the Vaidya collapse). This prompted physicists to seek a local definition of a black hole horizon, one which is dependent rather on its local geometry. This line of research was, in earnest, initiated by Hayward \cite{sh1}, who introduced the notion of a \textit{trapping horizon} (TH), which is characterized by the signs of the expansions of null geodesics to foliation surfaces (these are called \textit{marginally outer trapped surfaces} or MOTS for short) and their variations. (A generalization known as a \textit{marginally trapped tube} (MTT) was later introduced by Ashtekar and Galloway \cite{ag1}, where the ingoing null expansion is negative.) To remedy the problem of `\textit{globalness}' of the event horizon, the requirement of the existence of a global timelike vector was dropped in favor of the time independence of the local metric on the horizon. This is the local analogue of an event horizon, an \textit{isolated horizon} (IH), introduced by Ashtekar \textit{et al.} \cite{ak1}. That is, one can have a black hole in equilibrium in a dynamical spacetime, but away from the dynamical regime. In this case, the outgoing null expansion is constant along the null generator of the horizon. This allowed for the quasi-local formulation of the laws of black hole mechanics. A related notion is that of a \textit{dynamical horizon} (DH), which captures a black hole not in equilibrium, but one which grows in area by flux across the horizon \cite{ak2,ak3}. In this case, the outgoing null expansion decreases along both outgoing and ingoing null geodesics (see the references \cite{ak2,ak3,ib1,ib2}), and it is a particular case of a TH. This notion captures the full non-linearity of GR, and one which is expected in realistic gravitational collapse. 

A closely related notion to a DH is a stability characterization of the foliation surfaces of the DH, defined by the sign of the principal eigenvalue of some elliptic operator \cite{lars1}. It is expected that a DH enclosing a trapped region must be foliated by stable MOTS. It was recently discussed, \cite{she1}, that the conditions that the variations of the outgoing null expansion along the outgoing and ingoing null directions being negative is not sufficient, if the DH is to enclose a trapped region, but that their ratio must lie in the interval \((0,1)\). (This explains a result by Senoville \cite{sen1} that a DH may not enclose any trapped surface.) 

A problem with horizons associated with dynamical black holes is that some standard thermodynamic concepts do not directly carry over. For example, the notion of surface gravity is not well defined for a dynamical horizon, and therefore, there is no consensus on an acceptable definition for temperature in these cases. To this end, the notion of a \textit{conformal Killing horizon} has been suggested as the `right' boundary to associate to a black hole in a dynamical setting, provided there exists a conformal map from the dynamical spacetime to a static black hole spacetime \cite{niel1}. These objects are hypersurfaces on which the norm of an irrotational conformal Killing vector in the spacetime vanishes \cite{dy1,dy2}. The suggestion follows from the fact that the Killing horizon of the static black hole maps to the conformal Killing horizon of the dynamical black hole in such a way that, a particular definition of surface gravity, one used to define the Hawking temperature, is conformally invariant \cite{tj1}. 

While there is an extensive literature on horizons foliated by MOTS (for example, this has been extended to binary black hole mergers with the implementation of robust numerical codes, where unexpected behaviors of MOTS, like self-intersections, were observed \cite{ib3,ib4,ib5}), works on conformal Killing horizons is sparse, with relatively little attention being devoted to the topic. Conformal Killing horizons were introduced a long time ago by Dyer and Honig \cite{dy1}. However, it is (relatively) recently that the works by Jacobson and Kang \cite{tj1}, and Sultana and Dyer \cite{dy2} have shed further light on properties of these horizon types. Since then, a few other works have been devoted to studying these objects, albeit in different contexts. 

Chatterjee and Ghosh have recently introduced a Hamiltonian formulation of a conformal Killing horizon as well as the space of GR solutions admitting such objects \cite{gh1}. A rotation version of a conformal Killing horizon, foliated by distorted 2-spheres, was also introduced by the same authors \cite{gh2}, where a differential first law of black hole mechanics was defined on the horizon. A recent work by De Lorenzo and Perez, \cite{pere1}, have located bifurcate conformal Killing horizons in Minkowski spacetime, whose thermodynamic properties mimics those of black hole horizons in dynamical spacetimes. Recent works by Nielsen \cite{niel2} and Nielsen and Shoom \cite{niel1} have located the conformal Killing horizons in the pure Vaidya metric, as well as their thermodynamic property. This work was recently extended to the charged case, \cite{kps}, where the charge is time-dependent, where three conformal Killing horizons were located with two of them recovering the inner at outer horizons of the Reissner-Nordstrom metric in the static limit.

The aim of the current work is to investigate the criteria for the existence of conformal Killing horizon generated by a proper (non-homothetic) Killing vector field, in the class of locally rotationally symmetric (LRS) spacetimes \cite{ge1,ge2}. The appeal to this class of spacetimes stems from the fact that, many exact spherically symmetric dynamical GR solutions fall within this class. These spacetimes admit a multiply transitive isometry group, with a continuous isotropy group at each point of the spacetime. In local coordinates, the line element for the LRS class of spacetimes takes the form

\begin{eqnarray}\label{fork}
ds^2=-A^2dt^2+ B^2dr+ C^2dy^2+\left(\left(DC\right)^2+\left(Bh\right)^2-\left(Ag\right)^2\right)dz^2+2\left(A^2gdt-B^2hdr\right)dz,
\end{eqnarray}

where \(A,B\) and \(C\) are \(t\) and \(r\) dependent, and \(g\) and \(h\) are functions of \(y\) only. \(D\), which is a function of \(y\) is parametrized by a number \(k\in\{-1,0,1\}\): \(k=-1\) corresponds to \(\sinh y\), \(k=0\) corresponds to \(y\), \(k=1\) corresponds to \(\sin y\)). The particular case \(g=0=h\) is the subclass, LRS II spacetimes, which generalizes irrotational spherically symmetric solutions to the Einstein field equations. 

The existence of the preferred spatial direction \(e^a=B^{-1}\partial^a_r\) in the metric \eqref{fork} allows for the further decomposition of the 1+3 field equations of these spacetimes, the basis of the so-called 1+1+2 covariant formulation that will be employed in this work.

This paper has the following structure: In Section \ref{see2}, we briefly introduce the 1+1+2 covariant formalism that is used throughout this work. In Section \ref{ex}, we introduce conformal symmetries and conformal maps between spacetimes, and the notion of a conformal Killing horizon. We then establish the conditions for the existence of conformal Killing horizons in LRS spacetimes, and the constraint this imposes on the components of the conformal Killing vector generating the horizon in Section \ref{exist}. We conclude with discussion of the results and present a possible future research direction in Section \ref{dis}.

\section{The formalism}\label{see2}

There is by now a well established literature on the 1+1+2 covariant formalism, with several references providing details of the steps in the decomposition. A reader interested in more details is referred to the following standard texts \cite{cc1,cc2}. See the reference \cite{philip} which provides nice details of the formulation, and corrected for some of the resulting equations (a recent work by Hansraj and Goswami, \cite{rit1}, have also corrected some of the equations). As such, the introduction here would be very brief, and restricted to the LRS class of spacetimes.

In the well known and powerful 1+3 approach, a unit tangent direction - call this \(u^a\) - along which timelike observers flow, is what threads the spacetime. The field equations are decomposed along \(u^a\) plus some constraint equations. If there exists some unit spacelike vector \(e^a\) obeying \(u_ae^a=0\), the 3-space from the previous split can be further decomposed along this direction as a \(1+2\) product manifold. This, in addition to the \textit{evolution} equations along \(u^a\), introduces a set of \textit{propagation} equations along \(e^a\), as well as some constraint equations. The splitting naturally introduces new derivative along \(e^a\) and on the 2-space (generally referred to as the \textit{sheet}):

\begin{itemize}
\item \textit{Along \(u^a\)}: \\ \(\dot{\psi}^{a\cdots b}_{\ \ \ \ c\cdots d}=u^f\nabla_f\psi^{a\cdots b}_{\ \ \ \ c\cdots d}\);

\item \textit{Along \(e^a\)}: \\ \(\hat{\psi}^{a\cdots b}_{\ \ \ \ c\cdots d}=e^f\nabla_f\psi^{a\cdots b}_{\ \ \ \ c\cdots d}\);

\item \textit{Along the sheet}:\\ \(\delta_f\psi^{a\cdots b}_{\ \ \ \ c\cdots d}=N^{\bar c}_{\ c}\cdots N^{\bar d}_{\ d}N^a_{\ \bar a}\cdots N^b_{\ \bar b}N^e_{\ f}\nabla_e\psi^{\bar a\cdots \bar b}_{\ \ \ \ \bar c\cdots \bar d}\),
\end{itemize}
for any tensor \(\psi^{a\cdots b}_{\ \ \ \ c\cdots d}\), where \(N^{ab}=g^{ab}+u^au^b-e^ae^b\) projects vectors and tensors to the sheet (it is the metric induced from the splitting, on the sheet).

A spacetime vector \(\psi^a\) can accordingly be decomposed as \(\psi^a=\Psi u^a+\bar\Psi e^a\), and a fully projected tensor \(^3\psi_{ab}\) on the 3-space decomposes as

\begin{eqnarray}\label{dec}
^3\psi_{ab}=\Psi\left(e_ae_b-\frac{1}{2}N_{ab}\right),
\end{eqnarray}
Similarly, the gradient of scalars decomposes along these directions:

\begin{eqnarray}\label{decn}
\nabla_a\psi=-\dot{\psi}u_a+\hat{\psi}n_a.
\end{eqnarray}

The energy momentum and Ricci tensors take the covariant form

\begin{subequations}
\begin{align}
T_{ab}&=\rho u_au_b+ph_{ab}+2Qu_{(a}e_{b)}+\Pi\left(e_ae_b-\frac{1}{2}N_{ab}\right),\label{emt}\\
R_{ab}&=\frac{1}{2}\left(\rho-p+2\Lambda\right)g_{ab}+\left(\rho+p\right)u_au_b+\left(2Qu_{(a}e_{b)}+\Pi\left(e_ae_b-\frac{1}{2}N_{ab}\right)\right).\label{ric1}
\end{align}
\end{subequations}
Here, \(h_{ab}=g_{ab}+u_au_b\) is the induced metric on the 3-space (this introduces the derivative operator \(\bar D_a=h^b_a\nabla_b\)), \(\rho=T_{ab}u^au^b\) is the (local) energy density, \(3p=T_{ab}h^{ab}\) is the (isotropic) pressure, \(Q=T_{ab}e^au^b\) is the heat flux, \(\Pi=T_{ab}e^ae^b-p\)) captures the degree of anisotropy, and \(\Lambda\) is the cosmological constant.

The covariant derivatives of the unit vector fields can be written down in terms of the kinematic variables:

\begin{subequations}
\begin{align}
\nabla_au_b&=-\mathcal{A}u_ae_b+\left(\frac{1}{3}\theta+\sigma\right)e_ae_b+\frac{1}{2}\left(\frac{2}{3}\theta-\sigma\right)N_{ab}+\Omega\varepsilon_{ab},\label{cd1}\\
\nabla_ae_b&=-\mathcal{A}u_au_b+\left(\frac{1}{3}\theta+\sigma\right)e_au_b+\frac{1}{2}\phi N_{ab}+\bar\Omega\varepsilon_{ab}.\label{cd2}
\end{align}
\end{subequations}
\(\mathcal{A}=e_a\dot{u}^a\) is acceleration; \(\theta=h^{ab}\bar D_au_b\) is expansion of \(u^a\); \(\phi\) is expansion of \(e^a\) (called sheet expansion), \(\sigma=\sigma_{ab}e^ae^b\) is the shear; \(2\Omega=\varepsilon^{ab}\nabla_au_b\) and \(2\bar{\Omega}=\varepsilon^{ab}\nabla_ae_b\) are the respective twists of \(u^a\) and \(e^a\) (\(\Omega\) is usually referred to simply as vorticity or rotation); and \(\varepsilon_{ab}\) is the 2-dimensional alternating tensor.

In general, the dot and hat derivatives do not commute, but obeys the the relation

\begin{eqnarray}\label{a20}
\dot{\hat{\psi}}-\hat{\dot{\psi}}=-\mathcal{A}\dot{\psi}+\left(\frac{1}{3}\theta+\sigma\right)\hat{\psi},
\end{eqnarray}
for an arbitrary scalar \(\psi\).

The Einstein field equations takes a covariant form as a set of evolution and propagation equations of the covariant scalars, vectors and tensors, along with some constraint equations. These are obtained from the Ricci and Bianchi identities, taking components along the unit timelike and spatial directions, and projection onto the sheet. For the purpose of this work, we will not need to evoke the field equations here, and the derivation of the full set of equations can be found elsewhere (see \cite{cc1,cc2,philip} and references therein).
  

\section{Conformal Killing vectors and horizons}\label{ex}

This section introduces key concepts necessary to the development of the results herein. In particular, conformal symmetries in context of the LRS spacetimes, and conformal maps between static spacetimes and their conformal images are introduced.

A spacetime pair \(\left(M,g_{ab}\right)\) is said to admit a conformal Killing vector (CKV) \(\xi^a\) if Lie dragging the metric along \(\xi^a\) scales the metric:

\begin{eqnarray}\label{1b30}
\mathcal{L}_{\xi}g_{ab}=\nabla_a\xi_b+\nabla_b\xi_a=2\varphi g_{ab},
\end{eqnarray}
henceforth called the conformal Killing equations (CKE), where \(4\varphi=\nabla_a\xi^a\) . Depending on the character of the \textit{conformal divergence} \(\varphi\), the vector field \(\xi^a\) carries the particular nomenclature 

\begin{itemize}
\item a \textit{Killing vector} (KV) if \(\varphi=0\),
\item a \textit{homothetic Killing vector} (HKV) if \(\varphi\) is a non-zero constant,
\item a \textit{proper conformal Killing vector} - which will simply be abbreviated a (CKV) - if \(\varphi\) is non-constant.
\end{itemize}

Our interest here is in the case that the CKV is proper (we will later strengthen this by imposing that the gradient of the conformal divergence has no critical points, at least where the CKV becomes null). Additionally, we will required that the CKV is timelike.

A particular case of a CKV is the \textit{special conformal Killing vector} (SCKV). These have the property that \(\nabla_a\nabla_b\varphi=0\), and the gradient \(\nabla_a\varphi\) is nowhere vanishing.

\begin{definition}\label{def1}
A conformal Killiing horizon (CKH) in a spacetime admitting a timelike conformal Killing vector \(\xi^a\) is a null hypersurface on which the norms of both the vector field and its associated vorticity vector vanish. (See \cite{dy1,dy2}.)
\end{definition}

Let us briefly expand on the above definition. One begins with a spacetime that admits a conformal Killing vector. The existence of a conformal Killing Horizons is reliant on the notion of conformal observers (requiring that the vector field be timelike, with the conformal observers defined via the normalization \(\bar{\xi}^a=\left(1/\sqrt{-\xi^2}\right)\xi^a\), with \(\xi^2=\xi_a\xi^a\)) that travel along the orbits of the vector field. Spacetimes admitting timelike vector fields are said to be conformally stationary, i.e. they are conformal to a stationary spacetime. In the case that the conformal observers do not experience vorticity as they move along the vector field, the spacetime is conformal to a static one. (See \cite{dy1} for more discussions.)

Whenever these observers cannot be defined, one is led to the notion of a boundary, known as a (conformally) \textit{stationary limit 'surface'}, where the vector field is null. It was then established in \cite{dy1} that such hypersurface on which the norm of the conformal Killing vector vanishes, must also have the norm of its vorticity vanishing.

Of particular interest with regards to conformal Killing horizons is the case in which a dynamical spacetime is mapped to a static one, with horizons that are mapped in a conformal way. 

Consider a static spacetime \(M\) with metric \(g_{ab}\), and let \(\xi^a\) be a timelike Killing vector associated with \(g_{ab}\). Denote by \(\mathcal{T}\) the Killing horizon generated by \(\xi^a\), i.e.

\begin{eqnarray}\label{b1}
\xi_a\xi^a\overset{\mathcal{T}}{=}0,
\end{eqnarray}
with the hypersurface orthogonality condition holding

\begin{eqnarray}\label{b2}
\xi_{[a}\nabla_b\xi_{c]}\overset{\mathcal{T}}{=}0.
\end{eqnarray}
Let \(\tilde{g}_{ab}\) be a metric conformal to \(g_{ab}\):

\begin{eqnarray}\label{b3}
\tilde{g}_{ab}=\Xi^2g_{ab}\iff\tilde{g}^{ab}=\Xi^{-2}g^{ab},
\end{eqnarray}
where \(\Xi>0\) is a smooth function on \(M\). The manifold \(\left(M, \tilde{g}_{ab}\right)\) (or simply the metric \(\tilde{g}_{ab}\)) will henceforth be referred to as the conformal manifold. (Strictly speaking, the conformal manifold is the pair \(\left(M,[g]_{ab}\right)\) where \([g]_{ab}\) is an equivalence class of metrics on \(M\) related to each other through the scaling as in \eqref{b3}. Here we are just picking a representative of the equivalence class to refer to as the conformal manifold.) 

Indeed, we have that

\begin{eqnarray}\label{bk1}
\begin{split}
\mathcal{L}_{\xi}\tilde{g}_{ab}&=\tilde{\nabla}_a\tilde{\xi}_b+\tilde{\nabla}_b\tilde{\xi}_a\\
&=\left(\mathcal{L}_{\xi}\Xi^2\right)g_{ab}\\
&=\left(\mathcal{L}_{\xi}\Xi^2\right)\Xi^{-2}\tilde{g}_{ab}\\
&=2\tilde{\varphi}\tilde{g}_{ab},
\end{split}
\end{eqnarray}
where \(\tilde{\nabla}\) is the covariant derivative associated to the conformal manifold,

\begin{eqnarray}\label{bk2}
\tilde{\xi}_a=\tilde{g}_{ab}\xi^b=\Xi^2\xi_a,\qquad\mbox{(note that \(\tilde{\xi}^a=\xi^a\))}
\end{eqnarray}
which gives the relationship

\begin{eqnarray}\label{bk3}
\tilde{\varphi}=\mathcal{L}_{\xi}\left(\ln\Xi\right),
\end{eqnarray}
is the conformal divergence associated to \(\xi^a\) as a CKV for the conformal metric. Indeed the causal character of the vector field is preserved, i.e. \(\tilde{\xi}^a\) is also timelike.

It is clear that \(\xi^a\) is a KV for the conformal metric if and only if the conformal scale is constant with respect to the conformal observers so that the conformal factor vanishes. Otherwise, \(\xi^a\) is a CKV for the conformal metric. In this case, the transformation \eqref{b3} induces a map

\begin{eqnarray}\label{iden}
\mathcal{T}\mapsto\tilde{\mathcal{T}}
\end{eqnarray}
where \(\tilde{\mathcal{T}}\) is the hypersurface generated by the corresponding CKV \(\xi^a\) for the conformal metric:

\begin{eqnarray}\label{btee1}
\tilde{\xi}_a\tilde{\xi}^a\overset{\tilde{\mathcal{T}}}{=}0
\end{eqnarray}
The hypersurface orthogonality condition also follows:

\begin{eqnarray}\label{b4}
\tilde{\xi}_{[a}\tilde{\nabla}_b\tilde{\xi}_{c]}\overset{\tilde{\mathcal{T}}}{=}0.
\end{eqnarray}

The precise relationship between the static and dynamical spacetimes is determined by the character of the factor $\Xi$, which can be solved from \eqref{bk3}. Of course, retaining the black hole, one has to take care that $\Xi^2$ is unity at the horizon and is regular there. Furthermore, the map may not yield a solution to the Einstein field equations due to the fact that the equations are not invariant under conformal transformation.

Crucially, there apears to be a certain non-uniqueness of the map that is worth mentioning. As was recently discussed in \cite{kps}, one could in principle have a dynamical spacetime mapping to different static spacetimes due to the apparent arbitrariness of the conformal factor. In this particular case, there is a dependence, of the conformal factor, on an arbitrary function of the ratio of the radius to the time-dependent mass. Fixing this function to constant provided us with a desired static spacetime, yielding an invariant Hawking temperature. However, another non-constant choice of the function yields a (not desired) Euclidean signatured metric which is not asymptotically flat. 

Now, for our particular purpose with respect to the symmetries of LRS spacetimes, we will consider a CKV of the form 

\begin{eqnarray}\label{bk4}
\xi^a=\alpha u^a + \beta e^a,
\end{eqnarray}
with \(\alpha\neq0\) (wherever \(\alpha=0\), \(\xi^a\) is spacelike there). Its vorticity vector field can be written down explicitly as (see the reference \cite{kst1})

\begin{eqnarray}\label{fg1}
\begin{split}
\omega^a_{\xi}&=\ ^*\nabla^{[a}\xi^{b]}\xi_b\\
&=2\left(\alpha\Omega+\beta\bar{\Omega}\right)\xi_bu^{[b}e^{a]},
\end{split}
\end{eqnarray}
where the \(^*\) indicates the dual (the square bracket implies the usual antisymmetrization on the indices) and hence, for LRS II spacetimes, the vanishing of the norm of the vorticity is trivial. (Recall that \(\Omega\) and \(\bar{\Omega}\) are the respective vorticity scalars for the unit fields \(u^a\) and \(e^a\).) And, as our consideration is that of the class two solutions, the vanishing of the norm of a (timelike) CKV \(\xi^a\) is both necessary and sufficient to specify a CKH. More precisely, \textit{Suppose a LRS II spacetime \(M\) admit a timelike proper conformal Killing vector field \(\xi^a\). Then, \(M\) admits a CKH if an only if there exist null points in \(M\) for \(\xi^a\), in which case there exists a static black hole spacetime, \(M^*\), with a Killing horizon generated by \(\xi^a\).}

The CKE \eqref{1b30} can then be expressed in covariant form as the following set of four first order partial differential equations (obtained by fully projecting the CKE along the orthonormal unit directions, and the sheet)

\begin{subequations}
\begin{align}
\varphi&=\dot{\alpha}+\mathcal{A}\beta,\label{b17}\\
\varphi&=\hat{\beta}+\left(\frac{1}{3}\theta+\sigma\right)\alpha,\label{b18}\\
0&=\dot{\beta}-\hat{\alpha}+\mathcal{A}\alpha-\left(\frac{1}{3}\theta+\sigma\right)\beta,\label{b19}\\
2\varphi&=\alpha\left(\frac{2}{3}\theta-\sigma\right)+\beta\phi.\label{b20}
\end{align}
\end{subequations}
Because the equations are covariant, and we will work entirely in the conformal manifold, there is no ambiguity, and so we can drop the overhead tilde.

\section{Existence of conformal Killing horizons in LRS spacetimes}\label{exist}

In this section we investigate the existence of non-degenerate (non-vanishing surface gravity) conformal Killing horizons in LRS class II spacetimes, exploiting a thermodymic property of these horizons. To this end, we provide a result that will be useful in proving the main theorem.

Suppose a CKV \(\xi^a\) generates a CKH. Two (frequently used) definitions of the surface gravity associated with the horizon are \cite{tj1}

\begin{subequations}
\begin{align}
-\kappa_1 \xi_a&=\frac{1}{2}\nabla_{a} (\xi_b \xi^b),\label{ah1}\\
\kappa_2 \xi_a&=\xi^b\nabla_b\xi_a.\label{ah22}
\end{align}
\end{subequations}
(There are several definitions that are useful for different reasons, but of relevance to this work are the above two.) The two definitions are related via a \(\varphi\) translation:

\begin{eqnarray}
\begin{split}
-2\kappa_1\xi_a&=\nabla_{a}\left(\xi_b\xi^b\right)\\
&=2\xi^b\left(2\varphi g_{ab}-\nabla_b\xi_a\right)\\
&=2\left(2\varphi\xi_a-\xi^b\nabla_b\xi_a\right)\\
&=2\left(2\varphi-\kappa_2\right)\xi_a,
\end{split}
\end{eqnarray}
so that 

\begin{eqnarray*}
\kappa_1=\kappa_2-2\varphi. 
\end{eqnarray*}
We will now check how \(\kappa_1\) evolves along the conformal orbits. We establish this as the following statement:

\begin{lemma}\label{lem1} 
Let \(\xi^a\) be a timelike proper conformal Killing vector in a spacetime \(M\). Then, \(\xi^a\) generates a conformal Killing horizon with no homothetic points if and only if \(\kappa_1\) is constant along \(\xi^a\).
\end{lemma}

\begin{proof}
Taking the Lie derivative of \eqref{ah1} along \(\xi^a\), we have

\begin{eqnarray}\label{conka}
\begin{split}
-\mathcal{L}_{\xi}\left(\kappa_1\xi_a\right)&=\frac{1}{2}\mathcal{L}_{\xi}\left(\nabla_a\left(\xi_b\xi^b\right)\right)\\
&=\nabla_a\xi_b\mathcal{L}_{\xi}\xi^b+\xi^b\mathcal{L}_{\xi}\left(\nabla_a\xi_b\right)\\
&=\xi^b\mathcal{L}_{\xi}\left(\nabla_a\xi_b\right)\\
&=\xi^b\left(\nabla_a\xi^c\left(\nabla_c\xi_b+\nabla_b\xi_c\right)+\xi^c\nabla_c\nabla_a\xi_b\right)\\
&=\xi^b\left(\xi^c\nabla_c\nabla_a\xi_b+2\varphi\nabla_a\xi_b\right)\\
&=\xi^b\xi^c\nabla_c\nabla_a\xi_b-2\varphi\kappa_1\xi_a,
\end{split}
\end{eqnarray}
where from the second to the third line we have used  the fact that \(\mathcal{L}_{\xi}\xi^a=0\). Using the (easily checked) relation \(\mathcal{L}_{\xi}\xi_a=2\varphi\xi_a\), we have that \eqref{conka} becomes

\begin{eqnarray}\label{1ckvl2}
-\xi_a\mathcal{L}_{\xi}\kappa_1&=\xi^b\xi^c\nabla_c\nabla_a\xi_b
\end{eqnarray}
Combining the Ricci identities for \(\xi^a\) with the conformal Killing equations \eqref{1b30} we obtain the following relation for a conformal Killing vector (analogous to the KV Lemma \(\nabla_c\nabla_a\xi_b=R_{bacd}\xi^d\))

\begin{eqnarray}\label{ckvl}
\begin{split}
\nabla_c \nabla_a\xi_b &= R_{bacd}\xi^d + g_{ab}\nabla_c\varphi+g_{cb}\nabla_a\varphi-g_{ca}\nabla_b\varphi,
\end{split}
\end{eqnarray}
which we substitute in \eqref{1ckvl2} to obtain

\begin{eqnarray}\label{ckvl2}
-\xi_a\mathcal{L}_{\xi}\kappa_1=\xi_b\xi^b\nabla_a\varphi.
\end{eqnarray}
It therefore follows that

\begin{eqnarray}\label{ckvl3}
\mathcal{L}_{\xi}\kappa_1=0\iff\xi_b\xi^b=0,
\end{eqnarray}
provided there is no point at which \(\varphi\) is constant. That is, the existence of a CKH, given the conformal mapping \eqref{b3} between some static spacetime containing a Killing horizon with null generator \(\xi^a\), implies (and is implied by) the constancy of \(\kappa_1\) along the conformal Killing orbit.\qed
\end{proof}

This fact that \(\kappa_1\) is constant over the CKH, coupled with the fact that  it is conformaly invariant as shown in \cite{tj1}, has been used to argue that \(\kappa_1\) is the correct quantity to identify with the temperature of the CKH \cite{niel1}, as was mentioned in the introduction. For our purpose here, the thermodynamics is not to be dealt with. Of interest is the relation \eqref{ckvl3}.

\begin{remark}
\textit{In the above Lemma \ref{lem1} we explicitly specified that there are no homothetic points. In the case that we are considering horizons under the conformal map \(g_{ab}\rightarrow \Xi^2 g_{ab}\), as was earlier discussed, the existence of homothetic points is equivalent to the existence of a nontrivial solution to the following second order equation for \(\Xi\):}

\begin{eqnarray}\label{noh}
\left(\mathcal{L}_{\xi}^2-\varphi^2\right)\Xi=0,
\end{eqnarray}
\textit{where we have denoted \(\mathcal{L}_{\xi}^2\bullet=\mathcal{L}_{\xi}\left(\mathcal{L}_{\xi}\bullet\right)\) (can be obtained from \eqref{bk3}). Therefore, another way to ensure that there are no homothetic points is to impose that the equation \eqref{noh} admits no solution.}
\end{remark}

\subsection{Main results}

With the Lemma \ref{lem1}, we are now in the position to state and prove the main result.

\begin{theorem}\label{tht1}
Let \(M\) be a LRS II spacetime, and let (\ref{bk4}) be a (proper) CKV in \(M\), where \(\alpha,\beta\neq0\). Then, \(\xi^a\) generates a conformal Killing horizon if and only if either of its components, as well as the difference \(\dot{\alpha}-\dot{\beta}\), is constant along its orbits.
\end{theorem}
\begin{proof}
Let us begin by expanding the left hand side of \eqref{ah22} by inserting \eqref{bk4}:

\begin{eqnarray}\label{contd}
\begin{split}
\xi^b\nabla_b\xi_a&=\left(\alpha u^b+\beta e^b\right)\nabla_b\left(\alpha u_a+\beta e_a\right)\\
&=\varphi\xi_a+\left(\dot{\beta}+\alpha\mathcal{A}\right)\left(\beta u_a+\alpha e_a\right),
\end{split}
\end{eqnarray}
where we have used \eqref{b17}, \eqref{b18} and \eqref{b19} to substitute for \(\dot{\alpha},\hat{\beta}\) and \(\hat{\alpha}\) respectively. 

Now, in order to compare the right hand side of \eqref{contd} with the right hand side of \eqref{ah22}, it is reuired that we impose that the second term of \eqref{contd} is a scale of $\xi^a$. That is, we seek some (nonzero) smooth function \(X\) (we will see shortly why \(X\) has to be nontrivial), such that

\begin{eqnarray*}
X\xi^a=\left(\dot{\beta}+\alpha\mathcal{A}\right)\left(\beta u_a+\alpha e_a\right).
\end{eqnarray*}
This requires, by expanding the above and comparing component, that the function $X$ satisfies

\begin{subequations}
\begin{align}
X\alpha&=\beta\left(\dot{\beta}+\alpha\mathcal{A}\right),\label{contd1}\\
X\beta&=\alpha\left(\dot{\beta}+\alpha\mathcal{A}\right),\label{contd2}
\end{align}
\end{subequations}
in which case one has 

\begin{eqnarray}\label{contd3}
\kappa_2=\varphi+X.
\end{eqnarray}
This would then give the expression for \(\kappa_1\) as

\begin{eqnarray}\label{contd4}
\kappa_1=X-\varphi,
\end{eqnarray}
so that the constancy of \(\kappa_1\) becomes the following condition:

\begin{eqnarray}\label{contd5}
\mathcal{L}_{\xi}\left(X-\varphi\right)=0.
\end{eqnarray}
It is now clear that the trivial case \(X=0\) is ruled out since the horizon is required to have no homothetic points.

Now, by assumption, there are no homothetic points and hence, provided we can find such smooth function \(X\), the condition \eqref{contd5} is equivalent to the vanishing of the norm of \(\xi^a\), at which points \(\alpha=\beta\). We seek the criteria for the existence of the function \(X\).

Without imposing \eqref{contd5}, let us take a step back and assume the existence of $X$ such that we have $\kappa_2=\varphi+X$, and $X$ takes the form

\begin{eqnarray*}
X=\frac{B_1}{B_2}.
\end{eqnarray*}
We know that $X$ can have no zeros. We also know that $X$ is finite (actually a divergent $X$ gives a divergent $\omega_{\xi}^a$). Now, if $B_1$ vanishes at a point $\bar p$, $B_2$ must vanish at $\bar p$ in order for $X$ to be non-zero and defined there. Therefore, either $B_i=0$ or $B_i\neq0$, for $i=1,2$. If $B_i\neq0$, then, one finds, from \eqref{contd1} and \eqref{contd1}, that 

\begin{eqnarray*}
\left(\dot{\beta}+\alpha\mathcal{A}\right)\left(\alpha^2-\beta^2\right)=0.
\end{eqnarray*}
Indeed, $\dot{\beta}+\alpha\mathcal{A}\neq0$ for otherwise $X=0$. Hence, $\xi^a$ must be a null vector, $\xi^a=\alpha\left(u^a\pm e^a\right)$, contradicting that $\xi^a$ is timelike and only null on the horizon. We therefore must have that both $B_1$ and $B_2$ are simultaneously zero.

Taking the Lie derivative of \eqref{contd1} along \(\xi^a\), using \eqref{b17} and \eqref{contd5}, evaluating on the horizon (and noting that on the horizon \(\mathcal{L}_{\xi}\alpha=\mathcal{L}_{\xi}\beta\)), we have

\begin{eqnarray}\label{contd8}
X=-\frac{\alpha\mathcal{L}_{\xi}\left(\dot{\alpha}-\dot{\beta}\right)-\left(\dot{\beta}+\alpha\mathcal{A}\right)\mathcal{L}_{\xi}\alpha}{\mathcal{L}_{\xi}\alpha}.
\end{eqnarray}
And since both the numerator and denominator of \eqref{contd8} simultaneously vanish, the following conditions are thereby imposed

\begin{subequations}
\begin{align}
\mathcal{L}_{\xi}\alpha&=0,\label{contd10}\\
\mathcal{L}_{\xi}\left(\dot{\alpha}-\dot{\beta}\right)&=0.\label{contd10x}
\end{align}
\end{subequations}
In other words, the existence of \(X\) requires both \(\alpha\) (and therefore \(\beta\)), as well as \(\dot{\alpha}-\dot{\beta}\) to be constant along \(\xi^a\). And since \eqref{contd5} implies \(\xi_a\xi^a=0\), the conditions \eqref{contd10} and \eqref{contd10x} together imply the existence of a conformal Killing horizon.\qed 
\end{proof}

Notice that we do not utilize the vorticities of \(u^a\) and \(e^a\). Hence, one adequately assesses that the above result naturally extends to the non-LRS II case. It was shown in \cite{she2} that spacetimes of the non-LRS II types cannot admit dynamical black holes (those with horizon increasing in area), foliated by marginally outer trapped surfaces (MOTS). We also do know that there are indeed LRS II solutions that admit dynamical black holes (the Lemaitre-Tolman-Bondi solution is a case in point) that are bounded by dynamical horizons. There is, on the other hand, no explicit example of a non LRS II solution admitting a dynamical black hole, at least to our current knowledge. So, while conformal Killing horizons and dynamical horizons are structurally different objects as the former do not admit a foliation by MOTS, we take the result established in \cite{she2} as a potentially strong indication that these spacetimes may not admit conformal Killing horizons as well. However, it is also important to note that the existence of a dynamical horizon in a spacetime of the LRS II class does not necessarily imply that the spacetime admits a conformal Killing horizon. In this vein, our analyses here has been exclusively geared towards LRS II geometries.

Notice that by setting \(\alpha=\beta\) in \eqref{contd1} (or \eqref{contd2}) gives the explicit form of \(X\):

\begin{eqnarray}\label{contd6}
X=\dot{\beta}+\alpha\mathcal{A},
\end{eqnarray}
thereby giving the explicit form of the surface gravity as

\begin{eqnarray}\label{contd7}
\kappa_1=-\left(\dot{\alpha}-\dot{\beta}\right),
\end{eqnarray}
where we have used \eqref{b17} to substitute for \(\varphi\). This then gives \(\kappa_2=\dot{\alpha}+\dot{\beta}+2\alpha\mathcal{A}\). 

As can be checked, the form of the surface gravity could also directly be obtained using the formula (see for example \cite{tj1})

\begin{eqnarray*}
\begin{split}
\sqrt{2}\left(\kappa_1+\varphi\right)&=\sqrt{-\nabla_a\xi_b\nabla^{[a}\xi^{b]}}\\
&=\sqrt{\nabla_a\xi_b\nabla^b\xi^a-4\varphi^2},
\end{split}
\end{eqnarray*}
to obtain \eqref{contd7}. Our approach in the proof of Theorem \ref{tht1}, however, provides us with a relatively easier route to obtaining the surface gravity. 

The explicit form of \(\kappa_1\) suggests that the condition \eqref{contd10x} is sufficient, and that although easier to work with, the condition \eqref{contd10} is necessary but not sufficient.  However, the result of Theorem \ref{tht1} can be strengthened as follows:

\begin{theorem}\label{tht2}
Let \(M\) be a LRS II spacetime, and let (\ref{bk4}) be a (proper) CKV in \(M\), where \(\alpha\) and \(\beta\) are non-constant. Then, a necessary and sufficient condition for \(\xi^a\) to generate a conformal Killing horizon is that either of its components is constant along its orbits.
\end{theorem}
\begin{proof}
The necessity follows from Theorem \ref{tht1}. We establish the sufficiency. Suppose \(\mathcal{L}_{\xi}\alpha=0\). Then, using the CKE \eqref{b17} to \eqref{b20}, we obtain (after some detailed calculations)

\begin{subequations}
\begin{align}
\dot{\beta}+\hat{\beta}&=\left(\alpha-\beta\right)\left(\mathcal{A}+\frac{1}{3}\theta+\sigma\right)\label{fgd1}\\
0&=\beta\dot{\beta}+\alpha\hat{\beta}+\left(\alpha^2-\beta^2\right)\left(\frac{1}{3}\theta+\sigma\right).\label{fgd2}
\end{align}
\end{subequations}
Substituting (\ref{fgd2}) into (\ref{fgd1}), and using (\ref{b17}) and (\ref{b18}), we can simplify to obtain

\begin{eqnarray}
0=\left(\alpha-\beta\right)\dot{\alpha}.
\end{eqnarray}
By assumption, \(\alpha\) is not constant and hence, \(\alpha=\beta\), so that \(\xi^a\) is null where \(\mathcal{L}_{\xi}\alpha=0\).\qed
\end{proof}

The above result provides an alternative characterization of a CKH, with respect to the CKV considered in this work: CKH are critical points of component of a CKV along orbits of the CKV.

Lastly, we consider a result that follows from Theorem \ref{tht2}. Let us consider an arbitrary vector field \(\xi^a=\alpha u^a+ \beta e^a\) in a LRS II spacetime and suppose that \(\mathcal{L}_{\xi}\alpha=\alpha\dot{\alpha}+\beta\hat{\alpha}=0\). Multiply \eqref{b17} and \eqref{b19} respectively by \(\alpha\) and \(\beta\). Adding the resulting equations and using \eqref{b20} to substitute for \(\varphi\), we simplify to obtain

\begin{eqnarray}
\begin{split}
0&=\alpha\dot{\alpha}+\beta\hat{\alpha}\\
&=\frac{1}{3}\left(\alpha^2-\beta^2\right)\theta+\beta\left(\dot{\beta}+\frac{1}{2}\alpha\phi\right).
\end{split}
\end{eqnarray}
Indeed, if \(\theta\) is non-zero and it holds that

\begin{eqnarray}
\dot{\beta}+\frac{1}{2}\alpha\phi=0,
\end{eqnarray}
then, \(\alpha=\pm\beta\) at those points where \(\mathcal{L}_{\xi}\alpha=0\). In the `\(-\)' case, \(\xi\) aligns with the tangent vector to ingoing null rays to surfaces of constant \(t\) and \(r\). This therefore leads us to the following 

\begin{proposition}\label{proz}
Let \(\xi^a=\alpha u^a +\beta e^a\) be a timelike vector field in a LRS spacetime with non-vanishing expansion, which does not point along the tangent to ingoing null geodesics. If the components of the vector field satisfy 

\begin{eqnarray*}
\dot{\beta}+\frac{1}{2}\alpha\phi=0,
\end{eqnarray*}
\(\xi^a\) is a CKV and generates a CKH.
\end{proposition}

If one has an arbitrary vector field in an expanding LRS spacetime, Proposition \ref{proz} provides a mechanism to simultaneously check whether the vector field is a CKV and if it generates a CKH. 

By Proposition \ref{proz}, if the vector field is a CKV, then, under the accompanying assumptions, it necessarily generates a conformal Killing horizon. In this case the, again the Proposition might not apply to a non-LRS II case as we had earier argued that there is a strong reason to believe that a non-LRS spacetime cannot admit a conformal Killing horizon. 

\subsubsection*{Behavior of components of the CKV along the CKH}

Indeed, the existence of a conformal Killing horizon presents further constraints on the components of the CKV and their evolution. To see this, it can be checked that it follows that the conformal observers are converging, i.e. \(\varphi<0\), as one would expect in cases relevant to the formation of black holes. On the horizon, using \eqref{b20}, \(\varphi<0\) explicitly becomes

\begin{eqnarray}\label{contdz1}
\alpha\left(\frac{2}{3}\theta-\sigma+\phi\right)<0.
\end{eqnarray}
The term in the parenthesis is the expansion of outgoing null geodesics from 2-surfaces in the spacetimes under consideration (see the references \cite{rit2,she3}), and is non-zero here since the CKH is not a MOTS. The term cannot be negative either, as this would imply that 2-surfaces foliating the CKH lie to the interior of a MOTS, and hence, trapped. Therefore, the term must be positive, so that \(\alpha\) (and consequently \(\beta\)) is a (strictly) negative constant quantity on the CKH. 

Now, \(\mathcal{L}_{\xi}\alpha=\mathcal{L}_{\xi}\beta=0\implies \dot{\alpha}-\dot{\beta}=-\left(\hat{\alpha}-\hat{\beta}\right)\), thereby giving \(\kappa_1\) alternatively as

\begin{eqnarray}\label{contdz2}
\kappa_1=\hat{\alpha}-\hat{\beta}
\end{eqnarray}
It follows that 

\begin{eqnarray*}
\dot{\beta}>\dot{\alpha};\quad \hat{\beta}<\hat{\alpha}.
\end{eqnarray*}

If the acceleration \(\mathcal{A}\) is non-negative, we dismiss \(\dot{\beta}<0\): Since \(\hat{\beta}>0\), from \eqref{b18} (noting that \(\varphi<0\)) we have that \((1/3)\theta+\sigma>0\). We see that \eqref{b19} holds only if \(\mathcal{A}<0\). Thus, \(\beta\) must be increasing along \(u^a\). So, in the non-accelerating case, necessarily, \(\alpha\) decreases along \(u^a\). (We note that \((1/3)\theta+\sigma>0\) as well, so that existence of a CKH in this case rests on a balance between the shear and timelike expansion.)

From the result of Proposition \ref{proz}, it is immediately clear that, spacetimes with identically vanishing or strictly negative sheet expansions are ruled out since, in these cases, one would have \(\dot{\beta}\leq0\).


\subsection{On the special conformal Case}

We have established the necessary and sufficient condition for a CKV in LRS spacetime to generate a conformal Killing horizon. As we do not consider homothety, a natural first step away from homothety, which is proper, is a SCKV. By definition, the conformal divergence has a nowhere vanishing gradient, a crucial ingredient in the proof of Theorem \ref{tht1}. Here, we specialize to the case of a SCKV and establish the following:

\begin{proposition}\label{pro1}
A LRS spacetime cannot admit a CKH generated by a SCKV of the form \eqref{bk4}.
\end{proposition}

While the main purpose here is to check (non)-existence of a CKH in LRS spacetimes, we will not need to make use of the result of Theorem \ref{tht1}. Rather, we will make use of the geometric structure of the spacetimes itself. In verifying Proposition \ref{pro1}, we establish a result with regards to the existence of SCKV. Specifically, Proposition \ref{pro1} follows as a corollary to the following result.

\begin{proposition}\label{pro2}
A LRS spacetime cannot admit a SCKV of the form \eqref{bk4}, if the gradient of its associated conformal divergence is timelike.
\end{proposition} 

If the gradient of its associated conformal divergence is null, the spacetime is a generalized pp-wave spacetime , which is not LRS. And if spacelike, the admitted SCKV is necessarily spacelike (see \cite{coley1}). Hence, we are interested in the case for which the gradient is timelike.

It has long been known that the existence of SCKV in a spacetime, imposes geometric restrictions on the spacetime \cite{coley1,hall1}. It is these particular restrictions, related to the character of the Riemann curvature tensor, that we aim to use to establish Proposition \ref{pro2}. The result Proposition \ref{pro2} excludes a large class of spacetimes admitting SCKV, thereby buttressing the known fact of the restrictive nature of these vector fields. 

If a spacetime admits a SCKV \(\xi^a\) with \(\nabla_a\varphi\nabla^a\varphi\neq0\), one can always define a non-vanishing globally gradient vector field, with its associated potential function given by \cite{coley1}

\begin{eqnarray}
\nabla_af=F_{ab}\nabla^a\varphi;\quad f=\xi^a\nabla_a\varphi-\frac{1}{2}\varphi^2,
\end{eqnarray}
where

\begin{eqnarray}
\begin{split}
F_{ab}&=\nabla_{[a}\xi_{b]}\\
&=2Xu_{[a}e_{b]},
\end{split}
\end{eqnarray}
written explicitly in the 1+1+2 covariant form, with the function \(X\) as defined earlier. It is found that, in this case the Riemann and Ricci curvature tensors take the decomposed form \cite{hall1}

\begin{eqnarray*}
R_{abcd}=\mathcal{B}Z_{ab}Z_{cd};\qquad R_{ab}=\mathcal{B}Z^{c}_aZ_{cb},
\end{eqnarray*}
where \(\mathcal{B}\) is a scalar and

\begin{eqnarray*}
Z_{ab}=\  ^*\left(\nabla_a\varphi\nabla_bf-\nabla_b\varphi\nabla_af\right).
\end{eqnarray*}
Thus, the vanishing of \(\mathcal{B}\) imposes that the spacetime is necessarily flat. We explicitly compute

\begin{eqnarray}\label{rict}
R_{ab}=\frac{1}{4}\mathcal{B}\left(\dot{\varphi}\hat{f}-\hat{\varphi}\dot{f}\right)^2N_{ab}.
\end{eqnarray}
That is, surfaces of constant \(t\) and \(r\) fully encode the spacetime curvature. Hence, 

\begin{eqnarray*}
R_{ab}u^b=0=R_{ab}e^b.
\end{eqnarray*}
Comparing with the \(u^a\) and \(e^a\) components of \eqref{ric1}, we find that the following must be true:

\begin{eqnarray*}
\begin{split}
0&=Q,\\
0&=\rho+3p+2\Lambda,\\
0&=\rho-p+2\Pi+2\Lambda,
\end{split}
\end{eqnarray*}
so that \(\Pi=2p\). (The second constraint is the statement that the strong energy condition marginally holds.) Note that, from \eqref{rict}, \(N^{ab}R_{ab}=g^{ab}R_{ab}\). Then, applying this to \eqref{ric1} and using \(\Pi=2p\), imposes the vanishing of the cosmological constant. It then follows that the scalar \(\mathcal{B}\) can be expressed as

\begin{eqnarray}
\mathcal{B}=\frac{4\rho}{\left(\dot{\varphi}\hat{f}-\hat{\varphi}\dot{f}\right)^2},
\end{eqnarray}
so that the spacetime Riemann and Ricci curvature tensors take the simple forms

\begin{subequations}
\begin{align}
R_{abcd}&=\rho \varepsilon_{ab}\varepsilon_{cd},\label{fg0}\\
R_{ab}&=\rho N_{ab}.\label{fg1}
\end{align}
\end{subequations}

From the above calculations, it immediately follows that, if a LRS spacetime \(M\) is to admit a SCKV, the following must hold:

\begin{enumerate}

\item \textit{\(M\) has vanishing heat flux;}

\item \textit{\(M\) has vanishing cosmological constant;}

\item \textit{A non-flat \(M\) cannot be of a perfect fluid type;} 

\item \textit{A non-flat \(M\) cannot be pressure-free; and} 

\item \textit{The Riemann and Ricci curvatures of \(M\) are characterized entirely by the energy density. Thus, the only vacuum LRS solution that could possibility admit a SCKV is the Minkowski spacetime.}

\end{enumerate}

Statement \textit{3.} verifies, for LRS spacetimes, a result obtained by Coley and Tupper \cite{coley1}, which ruled out the existence of SCKV in perfect fluids.

Now, the admittance of a SCKV is equivalent to the symmetry condition

\begin{eqnarray}\label{fg2}
\mathcal{L}_{\xi}R_{ab}=0.
\end{eqnarray}
Substituting \eqref{fg1} into \eqref{fg2}, and taking the \(u^au^b\) component of the resulting tensorial equation, leads to the constraint

\begin{eqnarray}\label{fg3}
\rho\varphi=0.
\end{eqnarray}
The vacuum case, \(\rho=0\), can be ruled out as this will impose \(\nabla_af=0\), contradicting that this is a non-vanishing vector field. Thus, we must have \(\varphi=0\), which we have ruled out, and thereby establishing the result of Proposition \ref{pro2}.


\section{Discussion and outlook}\label{dis}

Whenever there is a conformal relationship between a static and a dynamical spacetime, the existence of a black hole in one implies the existence of a black hole in the other. The Killing horizon of the static spacetime maps conformally to a conformal Killing horizon in the dynamical spacetime. The laws of black hole mechanics for the static black hole is then mimicked by the dynamical black hole. 

In this work, the conformal invariance of these laws, specifically that of the constancy of the surface gravity of the horizon, has been used to obtain the criteria for a conformal Killing vector in a locally rotationally symmetric spacetimes to generate a conformal Killing horizon. It is established that, for a conformal Killing vector lying in the plane spanned by the preferred time and spatial directions, the vanishing of either of the components of the vector field is necessary and sufficient for the vector field to generate a conformal Killing horizon. The explicit form of the surface gravity is provided. The assumption that goes into the result is that the divergence of the vector field has no critical points. This provides and alternative characterization of a conformal Killing horizon as the set of critical points of the divergence of the vector field along its orbits. Employing this result, we identify a simple differential criterion on the components of an arbitrary vector field in a LRS II spacetime with non-vanishing expansion, which determines if the vector field is a conformal Killing vector and that it generates a conformal Killing horizon.

The character of the components of the CKV on the horizon and their behavior along the preferred time and spatial directions of the spacetime is elucidated. This provides us with immediate non-existence checks, depending on the signs of the dot and hat derivatives of the components of the vector field.

An immediate generalization away from a homothetic Killing vector is a special conformal Killing vector, which has been extensively studied in the literature. For these vector fields the antisymmetric part of their gradient is non-vanishing. Naturally, we were interested in whether a special conformal Killing vector generates a conformal Killing horizon. (In fact, the vanishing of the antisymmetric part of its gradient is seen to impose that on a conformal Killing horizon, the gradient of the divergence of the vector field is zero, which we have ruled out.) We answered this problem in the negative by demonstrating that, in fact, LRS spacetimes do not admit a special conformal Killing vector if the gradient of the divergence of the CKV is timelike. (If the gradient is null, the spacetime must be a generalized pp-wave, which is not LRS, and a spacelike gradient results in a spacelike special conformal Killing vector, which is not of interest to horizon formation.) The argument is purely geometric, and relies on certain restrictions on the geometry of spacetimes admitting certain vector field. Specifically, the Riemann curvature tensor can be written as a tensor product of bivectors, scaled by some regular function. In the LRS case, the curvature is completely characterized by the local energy density of the spacetime. We then make use of the fact that the SCKV is a symmetry of the Ricci tensor, imposing the constraint that a coupling of the energy density and the divergence of the SCKV must vanish, in which case, of course, the energy density has to vanish since we are not dealing with the Killing case. The vacuum case, however, is ruled out since this would imply the vanishing of the antisymmetric part of the gradient of the special conformal Killing vector, which is ruled out. This result further establishes the very `special' and restrictive nature of spacetimes admitting a special conformal Killing vector. 

As was mentioned in the introduction, a work by Chatterjee and Ghosh \cite{gh2} have considered conformal Killing whose foliation surfaces are distorted spheres. This was done via the introduction of rotation on the horizon. The 1+1+2 formulation used in this work extends to cases where the surfaces may not be exact spheres, and hence distorted. These general spacetimes can be realized as perturbations of LRS background solutions, and hence for future considerations, we aim to examine the criteria for the existence of rotating quasi-local conformal Killing horizons in a general 1+1+2 spacetime.

\section*{Acknowledgements}

The author gratefully acknowledges Miok Park of the IBS Center for Theoretical Physics of the Universe, for useful discussions during the early stages of this work. This research was supported by the Basic Science Research Program through the National Research Foundation of Korea (NRF) funded by the Ministry of education (grant numbers) (NRF-2022R1I1A1A01053784) and (NRF-2021R1A2C1005748).

\end{document}